\documentclass[runningheads]{llncs}

\newif\ifpublish
\publishtrue

\usepackage{cite}
\usepackage{amsmath}
\usepackage{amssymb}
\usepackage{amsfonts}
\usepackage{algorithm}
\usepackage[noend]{algpseudocode}
\usepackage{graphicx}
\usepackage{textcomp}
\usepackage{xcolor}
\usepackage{xspace}
\usepackage{cleveref}
\usepackage{pifont}
\usepackage{mathtools}
\usepackage{listings}

\usepackage{tikz}
\usetikzlibrary{arrows.meta}

\newcommand{\para}[1]{\vskip 1em\noindent\textbf{#1.}}

\renewcommand{\paragraph}{\para}
\newcommand{\revisit}[1]{\textcolor{blue}{#1}}

\newcommand{\sys}{HammerHead\xspace}
\newcommand{\GST}{\textsf{GST}}

\newcommand\StateX{\Statex\hspace{\algorithmicindent}}
\newcommand\StateXX{\StateX\hspace{\algorithmicindent}}
\newcommand{\alglinenoNew}[1]{\newcounter{ALG@line@#1}}
\newcommand{\alglinenoPop}[1]{\setcounter{ALG@line}{\value{ALG@line@#1}}}
\newcommand{\alglinenoPush}[1]{\setcounter{ALG@line@#1}{\value{ALG@line}}}



\definecolor{eclipseStrings}{RGB}{42,0.0,255}
\definecolor{eclipseKeywords}{RGB}{127,0,85}
\colorlet{numb}{magenta!60!black}

\lstdefinelanguage{json}{
  basicstyle=\normalfont\ttfamily,
  commentstyle=\color{eclipseStrings},
  stringstyle=\color{eclipseKeywords},
  stepnumber=1,
  numbersep=8pt,
  showstringspaces=false,
  breaklines=true,
  string=[s]{"}{"},
  comment=[l]{:\ "},
  morecomment=[l]{:"},
  literate=
    *{0}{{{\color{numb}0}}}{1}
    {1}{{{\color{numb}1}}}{1}
    {2}{{{\color{numb}2}}}{1}
    {3}{{{\color{numb}3}}}{1}
    {4}{{{\color{numb}4}}}{1}
    {5}{{{\color{numb}5}}}{1}
    {6}{{{\color{numb}6}}}{1}
    {7}{{{\color{numb}7}}}{1}
    {8}{{{\color{numb}8}}}{1}
    {9}{{{\color{numb}9}}}{1}
}

\newtheorem{observation}{Observation}
\newtheorem{clm}{Claim}

\begin{document}

\title{\sys: Leader Reputation for \\Dynamic Scheduling}

\ifpublish
    \author{
        Giorgos Tsimos\inst{1,2}\thanks{Work done when the author was an intern at MystenLabs.} \and 
        Anastasios Kichidis\inst{1} \and 
        Alberto Sonnino\inst{1,3} \and 
        Lefteris Kokoris-Kogias\inst{1,4}
    }
    \authorrunning{G. Tsimos, A. Kichidis, A. Sonnino, L. Kokoris-Kogias}
    \institute{
        MystenLabs \and
        University of Maryland \and 
        University College London \and
        IST Austria
    }
\fi

\maketitle

\begin{abstract}

    The need for high throughput and censorship resistance in blockchain technology has led to research on DAG-based consensus. The Sui blockchain protocol uses a variant of the Bullshark~\cite{spiegelman2022bullshark} consensus algorithm due to its lower latency, but this leader-based protocol causes performance issues when candidate leaders crash.
    In this paper, we explore the ideas pioneered by Carousel~\cite{cohen2022aware} on providing Leader-Utilization and present \sys. Unlike Carousel, which is built with a chained and pipelined consensus protocol in mind, \sys does not need to worry about chain quality as it is directly provided by the DAG, but needs to make sure that even though validators might commit blocks in different views the safety and liveness is preserved. Our implementation of \sys shows a slight performance increase in a faultless setting, and a drastic 2x latency reduction and up to 40\% throughput increase when suffering faults (100 validators, 33 faults).

\end{abstract}

\section{Introduction}
\label{sec:intro}


Advances in Byzantine Fault Tolerant State-Machine-Replication (SMR) dictated by the needs of blockchain technology to have high throughput and censorship resistance (also referred in the literature as Chain Quality~\cite{cohen2022aware}) has resulted in a surge of research around DAG-based consensus~\cite{danezis2022narwhal, keidar2022dagrider, spiegelman2022bullshark, spiegelman2023shoal, gao2022dumbo, yang2022dispersedledger, blackshear2023sui,malkhi2022maximal}. These protocols are now being deployed in production environments. For instance, Bullshark~\cite{spiegelman2022bullshark} has been adopted by the Sui blockchain~\cite{sui} and is on the roadmap of Aptos~\cite{aptos} due to its lower latency and non-reliance on setting up and maintaining a common coin. This, however, comes with the caveat that Bullshark is a leader-based protocol which results in performance deterioration when some candidate leaders inevitably crash or are taken down for maintenance and software update.

This phenomenon has been already seen in Sui's production deployment. For example on August 29th, between 15:30 and 17:30 UTC, suddenly 10\% of the validators started being less responsive. This resulted in the p95 latency going up from 3 seconds to 4.6 seconds and even the p50 latency increasing from 1.9 seconds to 2.2 seconds. This is especially alarming because at that point the system was under low load (only 130 tx/sec) so the lost capacity did not affect the latencies.
Furthermore, in real blockchains, validators vary in stake and thus leader election frequency. Some high-stake validators act as leaders more ofthen than others, but when they briefly fail or undergo maintenance, performance suffers, causing stress for node maintainers who must work tirelessly to restore them. This pressure arises because missing many leader spots affects overall performance. \sys eases this burden by promptly removing these major validators from the leader schedule temporarily and swiftly reintegrating them when they recover, ensuring seamless operation.
These findings confirm our intuition that the cost of not having a leader-aware SMR~\cite{cohen2022aware} is significant even in DAG-based consensus protocols.

To resolve this challenge we design a leader-aware SMR for DAG-based consensus protocols. Inspired by Carousel, we also rely on on-chain metrics to achieve high leader-utilization. However, doing this on a DAG instead of a chain is not trivial. Firstly, unlike chained consensus protocols, DAGs do not commit blocks in the same view for all nodes. As a result, we cannot simply rely on the consensus protocol for agreement but need to open the black-box and adapt the way the DAG is interpreted to get safety and liveness. Additionally, DAG-based consensus protocols provide chain quality by design. Hence we need not aggressively diversify on who is the leader.
The \sys protocol is run locally by each validator and does not require any extra protocol message or cryptographic tool.

To achieve Leader Utilization \sys relies on the classic parent-based voting scheme adopted by many DAG-based protocols (such as Bullshark~\cite{spiegelman2022bullshark}, Tusk~\cite{danezis2022narwhal}, Dag-Rider~\cite{keidar2022dagrider}, Fino~\cite{malkhi2022maximal}) to retrieve information regarding which parties are the fastest and most active during the current leader schedule. In every round, the fastest $2f+1$ parties to vote for a leader's proposal increase their respective scores. The scores accumulated during each schedule epoch are used during the calculation of the next leader schedule. Specifically the $f$ validators with the lowest score (corresponding to the least active validators, either due to crashes or Byzantine behavior) lose their schedule slots, which are allocated to the set of $f$ validators with the best score instead. A schedule change is triggered after a predetermined number of rounds has passed, but only upon an observable commit of the DAG in order to preserve the safety of the system.

\paragraph{Main challenges}
The main technical challenges lie with maintaining all the properties of Byzantine Atomic Broadcast, while also guaranteeing liveness. A major difference from static leader schedules is that now different honest validators might be operating under different schedules. What we show is that all honest validators eventually allocate the same interval of rounds to the same schedule and thus have agreement on the DAG and on the schedule changes.
\sys also solves different challenges than Carousel~\cite{cohen2022aware}. Carousel targets chained consensus protocols where the safety of the protocol is guaranteed even when honest validators disagree on the identity of leader: only liveness may suffer and need to be eventually restored. In contrast, \sys operates on DAG-based protocols where disagreement on the leader's identity may lead to safety violations.

\paragraph{Real-world system}
We provide a \emph{production-ready} and \emph{fully-featured} (crash-recovery, monitoring tools, etc) implementation of \sys that has been adopted by
\ifpublish
    the Sui blockchain: \sys runs within the Sui mainnet since version \texttt{mainnet-v1.9.1}~\footnote{\url{https://github.com/MystenLabs/sui/releases/tag/mainnet-v1.9.1}}.
\fi
Our evaluation shows that \sys (i) introduces no throughput loss and even
provides small latency gains when the protocol runs in a faultless setting, (ii) drastically improves both latency and throughput
in the presence of crash-faults, and that unlike Bullshark that deteriorates with more faults, \sys maintains performance (up to 2x latency reduction and 40\% throughput improvement for 100-validator deployments suffering 33 faults); and (iii) does not suffer from any visible throughput degradation despite crash-faults.

\paragraph{Contributions}
We make the following contributions:
\begin{itemize}
    \item We present \sys, the first\footnote{Developed concurrently with Shoal~\cite{spiegelman2023shoal}, see \Cref{sec:related}.} reputation-based leader-election mechanism for DAG-Based consensus protocols.
    \item We formally prove that \sys achieves Safety, Liveness, and Leader Utilization.
    \item We provide a production-ready and fully-featured implementation of \sys and demonstrates its benefits through extensive benchmarks.
\end{itemize}

\section{Preliminaries}
\label{sec:prelims}

\subsection{Model}
\paragraph{Network} We assume a set $\Pi$ of $n$ parties (or validators; both are used interchangeably throughout this work) $\{p_1,\dots, p_n\}$ and an \emph{adaptive} adversary $\mathcal{A}$ that can corrupt up to $f<n/3$ of the parties arbitrarily, at any point. A party is \emph{crashed} if it halts prematurely at some point during execution. Parties that deviate arbitrarily from the protocol are called \emph{Byzantine} or bad. Parties that are never crashed or Byzantine are called \emph{honest}.
Parties are communicating over a partially synchronous network~\cite{JACM:DwoLynSto}, in which there exists a special event called Global Stabilization Time ($\GST$) and a known finite time bound $\Delta$, such that any message sent by a party at time $x$ is guaranteed to arrive by time $\Delta+ \max\{\GST, x\}$. 

\paragraph{Threat model} 
The adversary is computationally bounded. Pairwise points of communication between any two honest parties are considered \emph{reliable}, i.e. any honest message is \emph{eventually} (after a finite, bounded number of steps) delivered.
However, until $\GST$ the adversary controls the delivery of all messages in the network, with the only limitation that the messages must be eventually delivered. After $\GST$, the network becomes synchronous, and messages are guaranteed to be delivered within $\Delta$ time after the time they are sent, potentially in an adversarially chosen order.  

\subsection{Building Blocks}
\sys leverages the \emph{reliable broadcast} primitive. 

\begin{definition} [Reliable Broadcast] \label{def:rb}
    Each party $P_i$ broadcasts messages by calling $r\_bcast_i(m,r)$, where $m$ is a message and $r\in\mathbb{N}$ is a round number. Each party $P_j$ outputs $r\_deliver_j(m,r,i)$, where $m$ is a message, $r$ is a round number, and $i\in[n]$ the index of party $P_i$ who called the corresponding $r\_bcast_i(m,r)$. A Reliable Broadcast protocol achieves the following properties:

\begin{itemize}
    \item[] \textbf{Agreement.} If an honest party $P_i$ outputs $r\_deliver_i(m,r,k)$, then all other honest parties $P_j$ eventually output $r\_deliver_j(m,r,k)$.
    \item[] \textbf{Integrity.} For every round $r\in\mathbb{N}$ and for every $k\in[n]$, an honest party $P_i$ outputs $r\_deliver_i(m,r,k)$ at most once, regardless of $m$.
    \item[] \textbf{Validity}. If an honest party $P_k$ calls $r\_bcast(m,r)$, then eventually every honest party $P_i$  outputs \allowbreak$r\_deliver_i(m,r,k).$ 
\end{itemize}
\end{definition}



\subsection{Problem Definition}

Our result focuses on achieving \emph{Byzantine Atomic Broadcast (BAB)}, while also satisfying additional properties. 
In order to keep notation clear between reliable and atomic broadcast, we refer to the BAB broadcast and deliver events as \emph{$a\_bcast(m,r)$} and \emph{$a\_deliver(m,r, p_j)$} respectively, where $m$ is some message, $r\in\mathbb{N}$ is a round number and $p_j$ is a party out of the $n$ parties. 

\begin{definition}[Byzantine Atomic Broadcast] \label{def:bab}
    Each party $P_i$ broadcasts messages by calling $a\_bcast_i(m,r)$, where $m$ is a message and $r\in\mathbb{N}$ is a round number. Each party $P_j$ outputs $a\_deliver_j(m,r,i)$, where $m$ is a message, $r$ a round number and $i\in[n]$ the index of party $P_i$ who called the corresponding $a\_bcast_i(m,r)$. 
    A \emph{Byzantine Atomic Broadcast} protocol satisfies the properties of Reliable Broadcast along with:
    \begin{itemize}
    \item[] \textbf{Total Order} If an honest validator $P_i$ outputs $\allowbreak a\_deliver_i\allowbreak(m,r,k)$ before $\allowbreak a\_deliver_i(m',r',k')$, then no honest party $P_j$ outputs $a\_deliver_j(m',r',k')$ before $\allowbreak a\_deliver_j(m,r,k)$.
\end{itemize}
\end{definition}

\noindent An additional property of interest to this work is \emph{Leader Utilization}, introduced in Spiegelman et al.~\cite{spiegelman2023shoal}. 

\begin{definition}[Leader-Utilization]\label{def:leader_util}
    A BAB protocol achieves \emph{Leader Utilization} if,  in crash-only executions, after GST, the number of rounds $r$ for which no honest party commits a vertex formed in $r$ is bounded.\\
\end{definition}

\section{The \sys Protocol}

We propose a protocol that satisfies both Safety and Liveness, while operating on a dynamically changing schedule of leaders. Our protocol is inspired by Carousel as far as how we identify the well-performing validators and giving them more chances of being leaders. Unlike Carousel, we need not worry about chain quality but we need to take extra steps to make sure that the protocol is safe and live although it is running over a DAG (see \Cref{sec:related} for a more detailed comparison).

The protocol starts with an initial schedule $S_0$, which is a fair round-robin unbiased of the results of the previous epoch. The schedule can be initialized by randomly permuting all validators based on their stake; For example, if each validator $u$ holds stake $\mathsf{stake}(u)$ and the total number of rounds that require leaders is $TR$, we initialize the schedule with each validator $u$ being the leader of $TR\times\mathsf{stake}(u)/\sum_{u}\mathsf{stake}(u)$ rounds in order and then randomly permute them.

To compute a new schedule $S'$ to switch from schedule $S$, we initialize a table $pos$ with all validators. In $pos$ there are two columns per validator, one with the initial number of slots they have on the previous schedule $S$ and another with the number of slots they will have on the schedule $S'$.
Each validator goes through all the rounds where $S$ is active and computes a data structure $\mathsf{scores}(\cdot)$ mapping each validator to their reputation score. Every validator starts with a reputation score of 0. 
Upon committing a sub-dag in Bullshark we update the reputation score of each validator, using some deterministic rule, in order to guarantee agreement across views. 
Since all validators observe the same sequence of committed sub-dags, they all attribute the same scores to validators. 

We propose the deterministic rule for updating reputation scores to be that \emph{each validator receives 1 point each time they vote for a leader’s proposal (i.e., there is a parent link from the block of the validator at round $r$ to the leader, according to schedule $S$, of round $r-1$).} 
The reputation score of each validator is increased by the number of points they accumulate.

The first subtle challenge to preserving Safety is that when we commit a sub-dag in Bullshark this happens through a subjective view of the DAG. This means that two validators might see a different subset of votes or they might even commit sub-dags at vastly different points in time. In order to resolve the first challenge we introduce a delay at the calculation of the reputation score. More specifically, although committing the leader is subjective what is consistent is that (a) every validator will eventually commit that same leader and (b) when the leader is committed the subDAG that gets committed is the same. Leveraging these two observations we calculate the reputation score up to but excluding the committed leader. 


Furthermore, we separate the execution of the BAB in schedule epochs, each of which lasts approximately $\mathsf{T}$ leaders\footnote{It might be slightly larger because the leader after the $\mathsf{T}$-th commit are crashed.}. Once the epoch ends the validators compute a new leaders' schedule $S'$ as follows:
They select a set $B$ that contains at most $f$ validators (by stake); this set contains the validators with the lowest reputation scores.
They also select a set $G$ of equal size to $B$ $(|G|=|B|)$; this set contains the validators with the highest reputation scores. Any ties for either of the sets are deterministically resolved.
The new schedule $S'$ is computed by round-robin replacing each $B$ validator with a $G$ validator from the previous schedule $S$.
To do the replacement we perform the following:
\begin{itemize}
    \item Pick a validator $P_b$ from $B$
    \item Find a slot they are leaders in $S$
    \item  Pick a validator $P_g$ from $G$
    \item Set $pos[v_g,1] \gets pos[v_g,1]+1 \; ; \; pos[v_b,1] \gets pos[v_b,1]-1$ and replace $P_b$ with $P_g$ in the new schedule $S'$
\end{itemize}
Once the $S'$ is calculated, the new schedule takes effect immediately.

The second and most critical challenge of \sys appears during the schedule switch. This is because validators may not commit a leader immediately, but through recursion over the DAG and after an unbounded number of rounds before GST. 
Nevertheless, we show that if we carefully apply the schedules through and induction and without skips we can avoid any Safety violations.

Finally, Liveness is also at risk as validators in Bullshark only wait to see the block proposal of the leader every time. However, if validators are not synchronized and each one has a different belief of who is the leader of round $r$ (because they are in a different schedule) then no leader might succeed in committing. An easy solution to this would be to forfeit responsiveness and make every round last $\Delta$. Fortunately, in \sys avoids this and creates a responsive protocol by opening up the Bullshark algorithm and ensuring that after GST all honest validators will be in sync (or the adversary will have to keep them out of sync by committing subdags, effectively providing Liveness as well).




\begin{algorithm*}[t]
    \caption{Data structures and basic utilities for party $p_i$}
    \label{alg:dataStructures}
    \begin{algorithmic}[1]
    \footnotesize

        \Statex \textbf{Local variables:}
        \StateX struct $\textit{vertex } \mathsf{v}$:
        \Comment{The struct of a vertex in the DAG}
        \StateXX $ \mathsf{v}.\textit{round}$ - the round of $ \mathsf{v}$ in the DAG
        \StateXX $ \mathsf{v}.\textit{source}$ - the party that broadcast $ \mathsf{v}$
        \StateXX $ \mathsf{v}.\textit{block}$ - a block of transactions information
        \StateXX $ \mathsf{v}.\textit{edges}$ - a set of at least $n-f$ vertices in
        $ \mathsf{v}.\textit{round}-1$
        \Comment{Provide fairness}
        \StateX $DAG_i[]$ - An array of sets of vertices
        \StateX \revisit{$\mathsf{activeSchedule}$} - auxiliary info related to the schedule change. Input to the deterministic \textsc{getLeader}($\cdot$) function.

        \Statex
        \Procedure{\textsc{path}}{$ \mathsf{v}, \mathsf{u}$} \Comment{Check if exists a path from $\mathsf{v}$ to $\mathsf{u}$ in the DAG}
        \State \Return exists a sequence of $ \mathsf{k} \in \mathbb{N}$,
        vertices $ \mathsf{v_1},  \mathsf{v_2},\ldots,  \mathsf{v_k}$  s.t.\
        \StateXX $ \mathsf{v_1} =  \mathsf{v} $, $ \mathsf{v_k} =  \mathsf{u}$, and $\forall j \in [2..k]
            \colon  \mathsf{v_j} \in \bigcup_{\mathsf{r} \geq 1}
            \mathsf{DAG_i[r]} \wedge \mathsf{v_j} \in  \mathsf{v_{j-1}}.\textit{edges} $
        \EndProcedure


        \Statex
        \Procedure{\textsc{getAnchor}}{$\mathsf{r}$}
        \State \revisit{ $\mathsf{p} \gets \textsc{getLeader}(\mathsf{r}, \mathsf{activeSchedule})$ }
        \Comment{Any public deterministic function}
        \If{$\exists \mathsf{v} \in \mathsf{DAG[r]}$ s.t.\ $\mathsf{v}.source = \mathsf{p}$}
        \State \Return $\mathsf{v}$
        \EndIf
        \State \Return $\bot$
        \EndProcedure

        \alglinenoNew{counter}
        \alglinenoPush{counter}

    \end{algorithmic}
\end{algorithm*}

\begin{algorithm}
\caption{\sys:  algorithm for party $p_i$.}
\begin{algorithmic}[1]
\alglinenoPop{counter}
\footnotesize

        \Statex \textbf{Local variables:}
      
        \StateX $\mathsf{orderedVertices} \gets \{\}$
        \StateX $\mathsf{lastOrderedRound} \gets 0$ \Comment{or $\mathsf{lastCommittedRound}$}
        \StateX $\mathsf{orderedAnchorsStack}  \gets $ initialize empty stack
    
       \Statex

        \Procedure{\textsc{TryCommitting}}{$\mathsf{v}$}
         \If{$\mathsf{v}.\textit{round} \text{ mod }2 = 1$  or $\mathsf{v}.\textit{round} = 0$}
            \State return 
        \EndIf    
    
        \State $\mathsf{anchor} \gets$ \textsc{getAnchor($\mathsf{v}.\textit{round-2}$)}
        \State $\mathsf{votes} \gets \mathsf{v}.\textit{edges}$
        \If{$|\{ \mathsf{v'} \in \mathsf{votes}: \textsc{path}(\mathsf{v'},\mathsf{anchor})\}|
        \geq f+1$}
        
        \State \textsc{orderAnchors($\mathsf{anchor}$)}
        
        \EndIf

    \EndProcedure
       
     \StateX

    \Procedure{\textsc{orderAnchors}}{$\mathsf{v}$}
        \State $\mathsf{anchor} \gets \mathsf{v}$
        \State $\mathsf{orderedAnchorsStack}.\textit{push}(\mathsf{anchor})$
        \State $\mathsf{r} \gets \mathsf{anchor}.\textit{round} -2$
        \While{$\mathsf{r} > \mathsf{lastOrderedRound}$}        \State $\mathsf{prevAnchor} \gets \textsc{getAnchor}(\mathsf{r})$
            
            \If{$\textsc{path}(\mathsf{anchor}, \mathsf{prevAnchor})$}
            \State $\mathsf{orderedAnchorsStack}.\textit{push}(\mathsf{prevAnchor})$
            \State $\mathsf{anchor} \gets \mathsf{prevAnchor}$ 
            
        \EndIf
        
        \State $\mathsf{r} \gets \mathsf{r} -2$
        \EndWhile
        
        \State $\mathsf{lastOrderedRound} \gets \mathsf{v}.\textit{round}$
        \State $\textsc{orderHistory()}$
        
    \EndProcedure
    
    \Statex
    
    \Procedure{\textsc{orderHistory()}}{}
        \While{$\neg \mathsf{orderedAnchorsStack}.\textit{isEmpty}()$} 
        \State $\mathsf{anchor} \gets \mathsf{orderedAnchorsStack}.\textit{pop}()$ 
        \State \revisit{$t \gets \mathsf{activeSchedule.initialRound}+\mathsf{T}$} \Comment{\revisit{$\mathsf{T}:$ schedule-change frequency}}

        \If{\revisit{$t\le\mathsf{anchor.round}$}} 
            \State\revisit{$\mathsf{activeSchedule}\gets\textsc{updateSchedule}(\mathsf{anchor})$}
            \State \Return
        \EndIf
        \State $\mathsf{verticesToOrder} \gets \{\mathsf{v} \in \bigcup_{\mathsf{r} > 0}
        \mathsf{DAG_i[r]} \mid \textsc{path}(\mathsf{anchor},\mathsf{v}) \wedge \mathsf{v} \not\in $
       $\mathsf{orderedVertices}\}$
          \For{$\textbf{every} ~\mathsf{v} \in \mathsf{verticesToOrder}$ in some deterministic order}
              \State \textbf{order} $\mathsf{v}$ \Comment{\revisit{output $a\_deliver_i(v.block, v.round, v.source)$}}
              \State $\mathsf{orderedVertices} \gets \mathsf{orderedVertices} \cup \{\mathsf{v}\}$
          \EndFor
          
        \EndWhile
        \EndProcedure  

        \StateX


        \Procedure{\textsc{updateSchedule}}{$\mathsf{v}$}
        
        \For{\revisit{all rounds from $\mathsf{activeSchedule.initialRound}$} up to \revisit{$\mathsf{v.round}$}}
        \State \revisit{Add 1 to each validator's $\mathsf{scores}(\cdot)$ that voted for previous round's leader}
        \State \revisit{Compute $\mathsf{schedule}$: the updated schedule according to $\mathsf{scores}(\cdot)$}
        \EndFor
        \State \Return \revisit{$\mathsf{schedule}$}
        \EndProcedure  

        \alglinenoPush{counter}

\end{algorithmic}
\label{alg:ESBullshark}

\end{algorithm}

\subsection{Protocol Specification} Our protocol can be seen in~\cref{alg:dataStructures} and \cref{alg:ESBullshark}. It operates on top of a DAG-based BAB protocol, such as Bullshark~\cite{CCS:SGSK22}. The main idea is to change the leader scheduling from static to adaptive, based on reputation scores. We have already explained how the scores are computed by each validator in our practical application. However, our solution is not specific to the calculation of the schedule and could work with any \emph{deterministic} schedule-change rule.  
There are also slight changes that our protocol incurs to the initial Bullshark protocol.

Specifically, since schedules are being updated after committing an anchor by observing vertices that voted for an anchor; this means that schedule changes may need to occur retroactively. As explained already, there may be cases where a validator was operating under a previous schedule for a few rounds, perhaps because they were unable to commit an anchor for some rounds. Once that validator commits a new anchor, they update their view accordingly and observe the new schedule. Thus, they need to retroactively apply the new schedule for the time-period that they where operating under the previous schedule, while the new schedule was active.



\subsection{Protocol Correctness}
We now prove the correctness of our construction. We show that \sys satisfies the properties of BAB, as well as Liveness and Leader Utilization.

Firstly, we can observe that apart from the schedule change, by construction, \sys executes the same as the eventually synchronous Bullshark. This means that if an epoch started with schedule $S_0$, until the first schedule change from the first honest validator, we can immediately derive all Bullshark properties. This is useful for the first part of the next 
\begin{observation}\label{obs:bullshark_per_schedule}
    Within the same schedule, \sys operates exactly as the eventually synchronous Bullshark protocol, and thus has the same properties. 
\end{observation}
    

We also have the following claims, which can be proven as in~\cite{keidar2022dagrider}.
\begin{clm}\label{clm:causal_history}
    When an honest party $P_i$ adds a vertex $u$ to its $DAG_i$, the entire causal history of $u$ is already in $DAG_i$.
\end{clm}

\begin{clm}\label{clm:vertex_agreement}
    If an honest party $P_i$ adds a vertex $u$ to its $DAG_i$ , then 
eventually all honest parties add $u$ to their DAG.
\end{clm}

We can observe that from these two claims, any two honest parties who commit vertex $u$, will have the same causal history for $u$.
\begin{observation}\label{obs:same_causal}
    If an honest party $P_i$ adds a vertex $u$ to its $DAG_i$ , then every honest party will (i) commit $u$ and (ii) upon committing $u$ will have the same causal history for $u$ in its $DAG$.
\end{observation}
\begin{proof}
    From Claim~\ref{clm:vertex_agreement} if $P_i$ commits $u$, then every honest party $P_j$ eventually commits $u$. From Claim~\ref{clm:causal_history}, $P_j$ will have the entire causal history of $u$ in $DAG_j$ upon committing $u$. \hfill\qed
\end{proof}



\begin{proposition}[Schedule Agreement]\label{proposition:schedule_agreement}
    Assume that all honest validators eventually switch every schedule according to the schedule switch rule. Then, if an honest validator $p_i$ switches to schedule $S$, eventually every honest validator will switch to schedule $S$. 
\end{proposition}

\begin{proof}
    Via strong induction. Base case: From $S_0$ to $S_1$; \\
    Let $S_0$ be the very first schedule of the epoch and assume that $P_i$ is the first honest validator who switches from $S_0$ to say $S_1$.
    According to Alg.~\ref{alg:ESBullshark}, $P_i$ must have committed some anchor for round $r_i\ge \mathsf{T}$, else the triggering of schedule switch would not occur.
    Say that another honest validator $P_j$ has so far committed up to round $r_j$, then $r_j<r$. If $r_j\ge r_i$, then according to Alg.~\ref{alg:ESBullshark} $P_j$ would have switched to the next schedule by $r_i$, which is also $S_1$ according to the view of $P_j$, from Observation~\ref{obs:same_causal} up to $r_i$.
    $P_j$ will commit some anchor $a_{r'_j}$ for some round $r'_j>r_i$. Then, from quorum intersection, since $P_i$ committed anchor say $a_{r_i}$ in round $r_i$, there will be a path from $a_{r'_j}$ to $a_{r_i}$. So, $P_j$ will order $a_{r_i}$, meaning that $P_j$ will switch schedules and from Observation~\ref{obs:same_causal}, it will switch to $S_1$.\\

    \noindent Assume that the statement holds for all schedules from $S_0$ up to $S_{k}$. We prove that this holds also for $S_{k+1}$.\\
    Let $P_i$ be the first honest validator who switches from $S_k$ to $S_{k+1}$.
    Then, for each other honest validator, who is in some schedule $S_r: r<k+1$, we can use the induction hypothesis, which means that each will switch to $S_k$ at some point. 
    According to Alg.~\ref{alg:ESBullshark}, $P_i$ must have committed some anchor, say $a_{r_i}$ for round $r_i\ge \mathsf{T}+ S_k.\mathsf{initialRound}$. 
    Say that another honest validator $P_j$ has so far committed up to round $r_j$, then $r_j<r_i$, for the same reason as in the base case.
    Eventually, $P_j$ will switch to $S_k$ and after that, $P_j$ will commit some anchor $a_{r'_j}$ for some round $r'_j>r_i$. Then, from quorum intersection, since $P_i$ committed $a_{r_i}$ in round $r_i$, there will be a path from $a_{r'_j}$ to $a_{r_i}$. So, $P_j$ will order $a_{r_i}$, which means that $P_j$ will switch schedules and, from Observation~\ref{obs:same_causal}, it will switch to $S_{k+1}$.
   \hfill\qed
\end{proof}


\begin{clm}
    Let $t$ be some timestep after $\GST$. If an honest party reliably brodacasts (or delivers) a message $m$ at time $t$, then all honest parties deliver $m$ by time $t+\Delta$.
\end{clm}

As also explained in~\cite{spiegelman2022bullshark}, this is satisfied, since before delivering the message, any honest party would multicast it to all other parties.

\begin{lemma}[View Synchronization]\label{lemma:view_synchronization}
    Let $t_{sync}= \GST + \Delta$. Let $S_{\max}$ be the latest schedule any honest party has advanced to before $\GST$. Then, by time $t_{sync}$, all honest parties can advance up to schedule $S_{\max}$.
\end{lemma}

\begin{proof}
    By time $t_{sync}$ all parties deliver all pre-$\GST$ messages. From Claim~\ref{clm:causal_history} and the fact that some honest party switched to schedule $S_{\max}$ before $\GST$, it is guaranteed that the causal histories of (\dots the anchors that upon commit, force switching to\dots) schedule $S_{\max}$ and all the intermediate schedules, are in $DAG_i$ for every honest party $P_i$.
\end{proof}

\begin{lemma}[View Distance]\label{lemma:view_distance}
After $\GST$, if an honest party enters schedule $S$ then all honest parties will be at some schedule $S' \geq S$ within $\Delta$ time.
\end{lemma}

\begin{proof}
  From reliable broadcast, if an honest party delivers sufficient messages to enter schedule $S$ it will broadcast this information to all honest parties, let's say wlog at time $t$ and t is after $\GST$. These messages will be delivered by all honest parties, the latest at $t+\Delta$. 
  An honest party will either ignore the messages because it is already at $S' \geq S$ or enter $S$.
\end{proof}

Now we will show Liveness in two cases. First we assume that there is no adversarial behaviour and show that honest parties will move from $S$ to $S+1$ in a bounded number of steps after GST. Then we will show that the only way for the adversary to prevent all parties from collectively advancing schedules is to keep some honest parties ahead. However, to keep those parties ahead the adversary will need to keep advancing schedules, providing Liveness as well.

\begin{lemma}[Schedule switch]\label{lemma:schedule_switch}
 Let $S$ be a schedule. After $\GST$, if all honest parties are in schedule $S$ after round $S.\mathsf{initialRound}+\mathsf{T}$, then all honest parties will switch to the next schedule.
\end{lemma}

\begin{proof}
    After $\GST$ honest parties are at most $\Delta$ away from each other (Lemma~\ref{lemma:view_distance}). Also, all parties are at schedule $S$. Then, within a bounded amount of time, the parties who are ahead, will be in round $\ge S.\mathsf{initialRound}+\mathsf{T}$. They either commit the new anchor and switch schedules, or they cannot. 
    If they switch, then by Lemma~\ref{lemma:view_distance} all honest parties will switch within $\Delta$. Else, within $\Delta$ all honest parties will be caught up and they will all be able to commit, so they all switch together. 
\end{proof}

\begin{lemma}[Liveness]\label{lemma:Liveness_responsive}
 Let $S_{\max}$ be the latest schedule any honest party has advanced to before $\GST$. Within a bounded number of steps some honest party will enter $S'=S_{\max}+1$
\end{lemma}

\begin{proof}
From view synchronization, every honest party will be at $S_{\max}$ at $\GST + \Delta$. Now there are two cases. First, if some honest party moves to $S_{\max}+1$ then by view distance all honest parties will move to $S_{\max}+1$ within $\Delta$. So Liveness is proven. Else all honest parties will be at $S_{\max}$ and from Bullshark Liveness will succesfully advance $\ge \mathsf{T}$ rounds. Thus, from Schedule switch, they will all switch to schedule $S_{\max}+1$. 
\end{proof}

\begin{lemma}[HammerHead BAB]\label{lemma:Safety_responsive} \sys satisfies Byzantine Atomic Broadcast per definition~\ref{def:bab}.
\end{lemma}
\begin{proof}
    Directly from Schedule Agreement, Liveness and Observation~\ref{obs:bullshark_per_schedule}.
\end{proof}

\begin{lemma}[Leader Utilization] \sys satisfies Leader Utilization per definition~\ref{def:leader_util}. Specifically, the number of rounds $r$ for which no honest party commits a vertex formed in $r$ is bounded by $O(\mathsf{T})\cdot f$.
\end{lemma}

\begin{proof}
After $\GST$, a crashed node will not cast votes. As a result from the calculation of reputation scores, it will be in the $B$ set the latest $O(\mathsf{T})$ rounds after it crashed and will not get in the $G$ set for as long as it is crashed. 
Therefore, the number of rounds for which no honest party commits a vertex is bounded by $O(\mathsf{T})$ for each of the up to $f$ crashed leaders.

\end{proof}

\section{Implementation} \label{sec:mplementation}
We implement a networked multi-core \sys validator in Rust forking the Narwhal-Bullshark implementation of Sui\footnote{\url{https://github.com/mystenlabs/sui}}. We select this codebase because it is the only production-ready implementation of a DAG-based consensus protocol deployed in the real world (at the best of our knowledge).
It uses Tokio\footnote{\url{https://tokio.rs}} for asynchronous networking, fastcrypto\footnote{\url{https://github.com/MystenLabs/fastcrypto}} for elliptic curve based  signatures. Data-structures are persisted using RocksDB\footnote{\url{https://rocksdb.org}}. We use QUIC\footnote{\url{https://github.com/quinn-rs/quinn}} to achieve reliable authenticated point-to-point channels.
By default, this Narwhal-Bullshark implementation uses traditional round-robin to elect leaders; we modify its leader election module to use \sys instead. Implementing our mechanism requires adding less than 600 LOC (+ 400 LOC of tests), and does not require any extra protocol message or cryptographic tool.
Contrarily to most prototypes, our implementation is \emph{production-ready} and \emph{fully-featured} (crash-recovery, monitoring tools, etc). It runs at the heart of 
\ifpublish
the Sui mainnet since version \texttt{mainnet-v1.9.1}~\footnote{\url{https://github.com/MystenLabs/sui/releases/tag/mainnet-v1.9.1}}.
\fi
We open source our implementation of \sys\footnote{
    \ifpublish
        \url{https://github.com/asonnino/sui/tree/hammerhead} (commit \texttt{03c96a3})
    \fi
}.
\section{Evaluation} \label{sec:evaluation}
We evaluate the throughput and latency of \sys through experiments on Amazon Web Services (AWS). We then show its improvements over the baseline round-robin leader-rotation mechanism of Bullshark~\cite{spiegelman2022bullshark}.
We aim to demonstrate the following claims.
\begin{itemize}
    \item \textbf{C1:} \sys introduces no throughput loss and even provides small latency gains when the protocol runs in ideal conditions (faultless setting).
    \item \textbf{C2:} \sys drastically improves latency and throughput in the presence of crash-faults; and its benefit increases with the number of faults.
    \item \textbf{C3:} \sys does not suffer from any visible throughput degradation despite (crash-)faulty validators. Note that evaluating BFT protocols in the presence of Byzantine faults is an open research question~\cite{twins}.
\end{itemize}

\para{Experimental setup}
We deploy our fully-featured \sys testbed on AWS, using \texttt{m5d.8xlarge} instances across 13 different AWS regions: N. Virginia (us-east-1), Oregon (us-west-2), Canada (ca-central-1), Frankfurt (eu-central-1), Ireland (eu-west-1), London (eu-west-2), Paris (eu-west-3), Stockholm (eu-north-1), Mumbai (ap-south-1), Singapore (ap-southeast-1), Sydney (ap-southeast-2), Tokyo (ap-northeast-1), and Seoul (ap-northeast-2). Validators are distributed across those regions as equally as possible. Each machine provides 10Gbps of bandwidth, 32 virtual CPUs (16 physical core) on a 2.5GHz, Intel Xeon Platinum 8175, 128GB memory, and runs Linux Ubuntu server 22.04. \sys persists all data on the NVMe drives provided by the machine (rather than the root partition). We select these machines because they provide decent performance and are in the price range of `commodity servers'.

In the following graphs, each data point is the average of the latency of all transactions of the run, and the error bars represent one standard deviation (errors bars are sometimes too small to be visible on the graph). We instantiate several geo-distributed benchmark clients submitting transactions at a fixed rate for a duration of 10 minutes; each benchmark client submits at most 350 tx/s and the number of clients thus depends on the desired input load.
The transactions processed by both systems are simple increments of a shared counter. The leader-reputation schedule is recomputed every 10 commits and excludes the 33\% less performant Validators
\ifpublish
    \footnote{Mainnet Sui uses more conservative parameters: it recomputes the schedule every 300 commits and only excludes the bottom 20\% of validators.}
\fi.
When referring to \emph{latency}, we mean the time elapsed from when the client submits the transaction to when it receives confirmation of the transaction's finality. When referring to \emph{throughput}, we mean the number of \emph{distinct} transactions over the entire duration of the run.

In addition to our codebase, we also open-source all orchestration and benchmarking scripts as well as measurements data\footnote{
    \ifpublish
        \url{https://github.com/asonnino/hammerhead-paper/tree/main/data}
    \fi
} to enable reproducible evaluation results.
Appendix~\ref{sec:tutorial} provides a tutorial to reproduce our experiments.

\para{Benchmark in ideal conditions}
\Cref{fig:no-faults} compares the performance of the baseline Bullshark and \sys running with 10, 50, and 100 honest validators.
Regardless of the committee size, the performance of Bullshark is similar to \sys. We observe a peak throughput around 4,000 tx/s (for committee sizes of 10 and 50) and 3,500 tx/s (for a committee size of 100) for both systems. The latency of Bullshark is slightly higher than \sys, at 3 seconds while \sys provides a latency of 2.7 seconds. This small latency gains is due to \sys's added benefit to focus on electing performant leaders. Leaders on more remote geo-locations that are typically slower are elected less often, the protocol is thus driven by the most performant parties.
These observations validate our claim \textbf{C1} stating that \sys introduces no throughput loss and provides small latency gains when the protocol runs in ideal conditions.

\begin{figure}[t]
    \centering
    \includegraphics[width=0.85\textwidth]{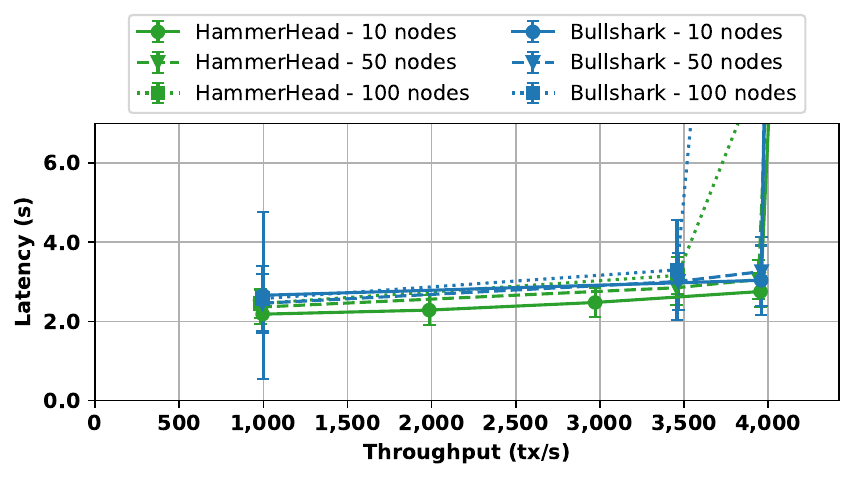}
    \caption{
        \sys and Bullshark latency-throughput performance with 10, 50, and 100 validators (no faults).
    }
    \label{fig:no-faults}
\end{figure}

\para{Benchmark with faults}
\Cref{fig:faults} compares the performance of Bullshark and \sys when a committee of 10, 50, and 100 validators respectively suffers 3, 16, and 33 crash-faults (the maximum that can be tolerated).

Bullshark suffers a massive degradation in both throughput and latency. For committee sizes of 10 and 50 suffering respectively 3 and 16 faults, the throughput of Bullshark drops by 25\% and its latency increases by 2-3x compared to ideal conditions.
In contrast, \sys only suffers a slight latency degradation (at most 0.5 second) due to a smaller pool of leaders to elect from. Notably, \sys does not suffer from any throughput degradation: it does not elect crashed leaders, the protocol continues to operate electing leaders from the remaining active parties, and is not overly affected by the faulty ones. This validates our claim \textbf{C2}.

The performance benefits of \sys are even more drastic for larger committees: for a committee size of 100 suffering 33 faults, the throughput of Bullshark drops by over 40\% and its latency increases 2x compared to ideal conditions. In contrast, \sys once again does not suffer from any throughput degradation and has only a slight latency increase.
We thus observe that \sys provides a 2x latency reduction and a  throughput increase ranging from 25\% (small committees) to 40\% (large committees) with respect to Bullshark. This validates our claim \textbf{C3}.

\begin{figure}[t]
    \centering
    \includegraphics[width=0.85\textwidth]{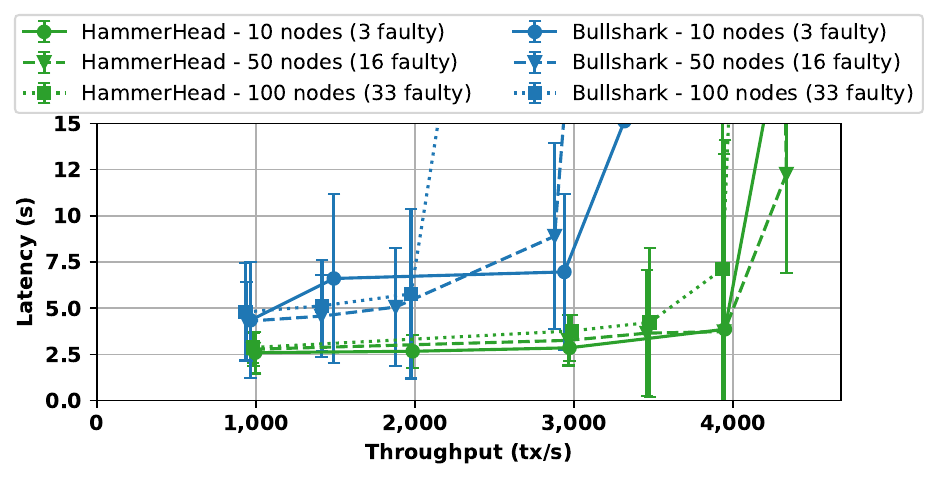}
    \caption{
        \sys and Bullshark performance with 10, 50, and 100 validators when experiencing their respective maximum number of tolerable faults. 
    }
    \label{fig:faults}
\end{figure}
\ifpublish
\section{Roadmap to Production} \label{sec:production}
We present an overview of our roadmap for the complete integration of \sys into the Sui mainnet. Despite the apparent simplicity of its algorithm, \sys brings substantial modifications to critical blockchain components, necessitating a comprehensive roadmap before its production deployment. This journey involved over four months of engineering effort by our team.

Our first milestone involved the implementation of the core reputation mechanism, which computes a reputation score for each validator and incorporates it into various non-critical aspects of the system, including a separate control system overseeing transaction submissions to consensus. This feature was seamlessly integrated ahead of the Sui mainnet launch.
Following this, we conducted extensive testing over several months to ensure that our chosen reputation metrics accurately reflected real-world performance. We have maintained continuous monitoring since the inception of the Sui mainnet, a period spanning approximately four months.
Next, we harnessed these reputation scores to fine-tune the leader schedule. This involved rigorous testing within our private deployments, followed by rigorous evaluation in the devnet and testnet environments, encompassing both load and failure testing. This phase consumed approximately 1.5 months.
With the confidence gained from successful private and test deployments, we made \sys publicly available and ran it for a month in the devnet and testnet environments. In parallel, it was integrated into the mainnet codebase, albeit gated through a protocol configuration and initially turned off.
Finally, we initiated the switch and incorporated \sys as a pivotal component of mainnet version 1.9.1\footnote{\url{https://github.com/MystenLabs/sui/releases/tag/mainnet-v1.9.1}}, corresponding to Sui protocol version 23, marking the culmination of this meticulous integration process.
\fi
\section{Related Work} \label{sec:related}
Carousel~\cite{cohen2022aware} presents the first reputation-based leader-rotations mechanisms for SMR protocols providing Leader Utilization. It specifically targets chained consensus protocols~\cite{buchman2016tendermint, yin2019hotstuff, baudet2019state, chan2020streamlet, gelashvili2022jolteon,chen2023resilient,cohen2022proof} and its main challenge thus lies in achieving Chain Quality~\cite{garay2015bitcoin}, which entails limiting the number of committed blocks proposed by Byzantine validators.
In contrast, \sys is tailored for DAG-based consensus protocols~\cite{danezis2022narwhal, keidar2022dagrider, spiegelman2022bullshark, spiegelman2023shoal, gao2022dumbo, yang2022dispersedledger, blackshear2023sui,malkhi2022maximal} and thus encounters distinct challenges. Unlike chained consensus protocols, DAG-based protocols do not preserve safety when validators disagree on the identity of the leader. As a result, \sys cannot simply
leverage the state of every view to recompute the reputation scores because different validators may commit the same block in different views. This distinction necessitates that we open the black-box of the DAG and adapt our interpretation of it to ensure both safety and liveness. On the positive side, using a DAG eliminates the need to be concerned about Chain Quality as \sys directly inherits it from underlying DAG, even if all leaders are malicious. Consequently, \sys forgoes the need to ensure that honest leaders make sufficiently frequent proposals.

One extreme scenario we also explored is that of the classic static leader that pre-blockchain BFT protocols used (e.g., PBFT~\cite{castro1999practical}), however, the risk of having a leader that performs just slow enough to not cause a gap in the schedule (and a subsequent ``schedule change'') is too great for the slight benefits of having a above average performance leader more often. We leave an open question if we can have a small subset of active leaders or a more adaptive reputation scoring mechanism to exploit the most performant nodes as leaders more often.

The Shoal framework~\cite{spiegelman2023shoal} (concurrent work, first appeared on ArXiv on June 2023) is the closest system to \sys. Shoal's primary objective is to lower the latency associated with DAG-based consensus, employing various strategies that include a leader-reputation mechanism like \sys.
Similar to \sys, Shoal's leader-reputation mechanism maintains a record of scores for each validator and employs a deterministic rule to recalibrate the mapping from rounds to leaders based on these scores. Shoal conceptually leaves open the choice of this deterministic rule and its implementation assigns higher scores to committed leaders and lower scores to leaders that were skipped. Conversely, \sys assigns scores based on the frequency of votes for leaders, discouraging Byzantine actors from withholding their votes for honest leaders.
Shoal and \sys however mostly diverge in their areas of emphasis. Shoal takes a broader perspective, focusing on reducing the latency of DAG-based consensus through additional techniques like consensus pipelining and prevalent responsiveness~\cite{spiegelman2023shoal}, while \sys entirely focuses on leader-reputation, offering detailed algorithms and formal security proofs.

\section{Conclusions}
\label{sec:conclusion}

This paper introduces \sys, a novel leader-aware SMR custom-designed for DAG-based consensus protocols. Drawing inspiration from Carousel and harnessing on-chain metrics, \sys achieves high leader utilization. To achieve this it  addresses the unique challenges posed by DAG structures, where block commitments lack synchronization across all nodes, by reinterpreting the DAG to ensure both safety and liveness. \sys's dynamic leader schedule adjustment, based on validator activity and reliability, optimizes leader selection while preserving system safety. This approach ensures sustained performance and throughput even in the presence of crash faults, outperforming existing leader-based protocols like Bullshark. In summary, \sys's implementation showcases its robustness in various scenarios, emphasizing the critical importance of leader-awareness in such systems.

\ifpublish
    \section*{Acknowledgements}
This work is supported by Mysten Labs. We thank the Mysten Labs Engineering teams for valuable feedback broadly, and specifically to Laura Makdah for helping implementing the early reputation score system for validators and Dmitry Perelman for managing the overall implementation effort.
\fi

\bibliographystyle{plain}
\bibliography{bib,cryptobib/abbrev3.bib,cryptobib/crypto}

\appendix
\section{Reproducing Experiments} \label{sec:tutorial}

We provide the orchestration scripts
\footnote{
\ifpublish
\url{https://github.com/asonnino/sui/tree/hammerhead} (commit \texttt{03c96a3})
\fi
} used to benchmark the codebase evaluated in this paper on AWS .

\para{Deploying a testbed}
The file `\texttildelow/.aws/credentials' should have the following content:
\begin{footnotesize}\begin{verbatim}
[default]
aws_access_key_id = YOUR_ACCESS_KEY_ID
aws_secret_access_key = YOUR_SECRET_ACCESS_KEY
\end{verbatim}\end{footnotesize}
configured with account-specific AWS \emph{access key id} and \emph{secret access key}. It is advise to not specify any AWS region as the orchestration scripts need to handle multiple regions programmatically.

A file `settings.json' contains all the configuration parameters for the testbed deployment. We run the experiments of \Cref{sec:evaluation} with the following settings:

\begin{footnotesize}
    \begin{lstlisting}[language=json]
{
    "testbed_id": "${USER}-hammerhead",
    "cloud_provider": "aws",
    "token_file": "/Users/${USER}/.aws/credentials",
    "ssh_private_key_file": "/Users/${USER}/.ssh/aws",
    "regions": [
        "us-east-1",
        "us-west-2",
        "ca-central-1",
        "eu-central-1",
        "ap-northeast-1",
        "ap-northeast-2",
        "eu-west-1",
        "eu-west-2",
        "eu-west-3",
        "eu-north-1",
        "ap-south-1",
        "ap-southeast-1",
        "ap-southeast-2"
    ],
    "specs": "m5d.8xlarge",
    "repository": {
        "url": "https://github.com/AUTHOR/REPO.git",
        "commit": "hammerhead"
    }
}
\end{lstlisting}
\end{footnotesize}

where the file `/Users/\${USER}/.ssh/aws' holds the ssh private key used to access the AWS instances, and `AUTHOR' and `REPO' are respectively the GitHub username and repository name of the codebase to benchmark.

The orchestrator binary provides various functionalities for creating, starting, stopping, and destroying instances. For instance, the following command to boots 2 instances per region (if the settings file specifies 13 regions, as shown in the example above, a total of 26 instances will be created):
\begin{footnotesize}\begin{verbatim}
cargo run --bin orchestrator -- testbed deploy --instances 2
\end{verbatim}\end{footnotesize}
The following command displays he current status of the testbed instances
\begin{footnotesize}\begin{verbatim}
cargo run --bin orchestrator testbed status
\end{verbatim}\end{footnotesize}
Instances listed with a green number are available and ready for use and instances listed with a red number are stopped. It is necessary to boot at least one instance per load generator, one instance per validator, and one additional instance for monitoring purposes (see below).
The following commands respectively start and stop instances:
\begin{footnotesize}\begin{verbatim}
cargo run --bin orchestrator -- testbed start
cargo run --bin orchestrator -- testbed stop
\end{verbatim}\end{footnotesize}
It is advised to always stop machines when unused to avoid incurring in unnecessary costs.

\para{Running Benchmarks}
Running benchmarks involves installing the specified version of the codebase on all remote machines and running one validator and one load generator per instance. For example, the following command benchmarks a committee of 100 validators (none faulty) under a constant load of 1,000 tx/s for 10 minutes (default), using 3 load generators:
\begin{footnotesize}\begin{verbatim}
cargo run --bin orchestrator -- benchmark \
    --committee 100 fixed-load --loads 1000 \
    --dedicated-clients 3 --faults 0 
    --benchmark-type 100
\end{verbatim}\end{footnotesize}
The parameter \texttt{benchmark-type} is set to 100 to instruct the load generators to sequence all transactions through the consensus engine. We select the number of load generators by ensuring that each individual load generator produces no more than 350 tx/s (as they may quickly become the bottleneck).

\para{Monitoring}
The orchestrator provides facilities to monitor metrics. It deploys a Prometheus instance and a Grafana instance on a dedicated remote machine. Grafana is then available on the address printed on stdout when running benchmarks with the default username and password both set to \texttt{admin}. An example Grafana dashboard can be found in the file `grafana-dashboard.json'\footnote{
    \ifpublish
        \url{https://github.com/asonnino/sui/blob/hammerhead/crates/orchestrator/assets/grafana-dashboard.json}
    \fi
}.

\para{Troubleshooting}
The main cause of troubles comes from the genesis. Prior to the benchmark phase, each load generator creates a large number of gas object later used to pay for the benchmark transactions. This operation may fail if there are not enough genesis gas objects to subdivide or if the total system gas limit is exceeded. As a result, it may be helpful to increase the number of genesis gas objects per validator in the `genesis\_config' file\footnote{
    \ifpublish
        \url{https://github.com/asonnino/sui/blob/03c96a3648f40f89bd78930b837aa1393bab73ec/crates/sui-swarm-config/src/genesis\_config.rs\#L360}
    \fi
} when running with very small committee sizes (such as 10).

\end{document}